\documentclass[submission,copyright,creativecommons]{eptcs}
\providecommand{\event}{AFL 2017} 
\usepackage{breakurl}             
\usepackage{underscore}           
\usepackage{tikz}
\usepackage{tikz-qtree}
\usepackage{adjustbox}
\usepackage{amsmath,amssymb,amsthm}
\usetikzlibrary{shapes.geometric,shapes.misc}
\usetikzlibrary{calc,shapes,arrows,shapes.misc,decorations.pathreplacing}
\usetikzlibrary{positioning,calc}
\usetikzlibrary{automata}
\usepackage{subcaption}
\usepackage{enumitem}

\theoremstyle{plain} 
\newtheorem{theorem}{Theorem} 
\newtheorem{lemma}{Lemma}

\theoremstyle{definition} 
{}

\theoremstyle{remark}

\title{Reversible Languages Having Finitely Many Reduced Automata}
\author{Kitti Gelle
\institute{University of Szeged\\Szeged, Hungary}
\email{kgelle@inf.u-szeged.hu}
\and
Szabolcs Iv\'an
\institute{University of Szeged\\Szeged, Hungary\thanks{Research supported by NKFI Grant no. 108448.}}
\email{szabivan@inf.u-szeged.hu}
}
\def\titlerunning{Reversible Languages Having Finitely Many Reduced Automata}
\def\authorrunning{K. Gelle \& Sz. Iv\'an}
\begin{document}
\maketitle

\begin{abstract}
Reversible forms of computations are often interesting from an energy efficiency point of view. When
the computation device in question is an automaton, it is known that the minimal reversible automaton
recognizing a given language is not necessarily unique, moreover, there are languages having
arbitrarily large reversible recognizers possessing no nontrivial “reversible” congruence. However,
the exact characterization of this class of languages was open. In this paper we give a forbidden
pattern capturing the reversible regular languages having only finitely many reduced reversible automata,
allowing an efficient (NL) decision procedure.
\end{abstract}

\section{Introduction}

Landauer's principle~\cite{landauer} states that any logically irreversible manipulation of information
-- such as the merging of two computation paths -- is accompanied by a corresponding
entropy increase of non-information-bearing degrees of freedom in the processing apparatus.
In practice, this can be read as ``merging two computation paths generates heat'', though
it has been demonstrated~\cite{vaccaro} that the entropy cost can be taken in e.g. angular momentum.
Being a principle in physics, there is some debate regarding its validity,
challenged~\cite{earmannorton,norton2005,norton2011}
and defended~\cite{bub,bennett,Ladyman200758} several times recently.
In any case, the study of \emph{reversible} computations, in which distinct computation paths never
merge, looks appealing. In the context of quantum computing, one allows only reversible logic gates~\cite{chuang}.
For classical Turing machines, it is known that each deterministic machine can be simulated by 
a reversible one, using the same amount of space~\cite{bennett72,langetapp}.
Hence, each regular language is accepted by a reversible \emph{two-way} deterministic finite automaton
(also shown in~\cite{kondacs}).

In the case of classical, i.e. \emph{one-way} automata, the situation is different:
not all regular languages can be recognized by a reversible automaton, not even if
we allow partial automata (that is, trap states can be removed, thus the transition
function being a partial one). Those languages that can be recognized by a reversible
one are called \emph{reversible languages}. It is clear that one has to allow being
partially defined at least since otherwise exactly the regular group languages
(those languages in whose minimal automata each letter induces a permutation)
would be reversible.

Several variants of reversible automata were defined and studied~\cite{pin92,angluin,lombardy}.
The variant we work with (partial deterministic automata with a single initial state and 
and arbitrary set of final states) have been treated in~\cite{kutrib14,kutrib15,holzer15,giovanni,giovanna,luca}.
In particular, in~\cite{holzer15} the class of reversible languages is characterized by means
of a forbidden pattern in the minimal automaton of the language in question, and an algorithm
is provided to compute a minimal reversible automaton for a reversible language, given its
minimal automaton. Here ``minimal'' means minimizing the number of states, and the minimal
reversible automaton is shown to be not necessarily unique.
In~\cite{giovanni}, the notion of \emph{reduced} reversible automata is introduced: a
reversible automaton is reduced if it is trim (all of its states are accessible and coaccessible),
and none of its nontrivial factor automata is reversible.
The authors characterize the class of those reversible languages (again, by developing a forbidden pattern)
having a unique reduced reversible automaton (up to isomorphism), and leave open the
problem to find a characterization of the class of those reversible languages
having finitely many reduced reversible automata (up to isomorphism).

In this paper we solve this open problem of~\cite{giovanni}, by also
developing a forbidden pattern
characterization which allows an $\mathbf{NL}$ algorithm.

\section{Notation}

We assume the reader has some knowledge in automata
and formal language theory (see e.g.~\cite{hopcroft}).

In this paper we consider deterministic \emph{partial} automata with a single initial state
and an arbitrary set of final states. 
That is, an \emph{automaton}, or DFA, is a tuple $M=(Q,\Sigma,\delta,q_0,F)$ with
$Q$ being the finite \emph{set of states},
$q_0\in Q$ the \emph{initial} state,
$F\subseteq Q$ the set of \emph{final} or \emph{accepting} states,
$\Sigma$ the finite, nonempty \emph{input alphabet} of symbols or \emph{letters}
and $\delta:Q\times\Sigma\to Q$ the \emph{partial} transition function
which is extended in the usual way to a partial function also
denoted by $\delta: Q\times\Sigma^*\to Q$ with
$\delta(q,\varepsilon)=q$ for the \emph{empty word} $\varepsilon$
and $\delta(q,wa)=\delta(\delta(q,w),a)$ if $\delta(q,w)$ is defined and
undefined otherwise.
When $\delta$ is understood from the context, we write $q\cdot w$ or $qw$ for $\delta(q,w)$
in order to ease notation.
When $M$ is an automaton and $q$ is a state of $M$, then \emph{the language recognized by $M$ from $q$}
is $L(M,q)=\{w\in\Sigma^*:qw\in F\}$.
The \emph{language recognized by $M$} is $L(M) = L(M,q_0)$.
A language is called \emph{regular} or \emph{recognizable} if some automaton recognizes it.

When $p$ and $q$ are states of the automata $M$ and $N$, respectively, we say that
$p$ and $q$ are \emph{equivalent}, denoted $p\equiv q$, if $L(M,p)=L(N,q)$.
(When $M$ or $N$ is unclear from the context, we may write $(M,p)\equiv(N,q)$.)
The automata $M$ and $N$ are said to be equivalent if their initial states are equivalent.

A state $q$ of $M$ is \emph{useful} if it is \emph{reachable}
($q_0w = q$ for some $w$) and \emph{productive} ($qw \in F$ for some $w$).
A DFA is \emph{trim} if it only has useful states.
When $L(M)$ is nonempty, one can erase the non-useful states of $M$:
the resulting automaton will be trim and equivalent to $M$
(and may be partially defined even if $M$ is totally defined, if $M$ has a trap state $q$ for which $L(M,q)=\emptyset$).
An equivalence relation $\Theta$ on the state set $Q$ is a \emph{congruence} of $M$ if
$p\Theta q$ implies both $p\in F~\Leftrightarrow~q\in F$ (that is, $F$ is \emph{saturated} by $\Theta$)
and $pa\Theta qa$ for each $a\in\Sigma$ (that is, $\Theta$ is \emph{compatible} with the action).
In particular, in any $\Theta$-class,
$pa$ is defined if and only if so is $qa$.

Clearly the identity relation $\Delta_Q$ on $Q$ is always a congruence, the \emph{trivial} congruence.
A trim automaton is \emph{reduced} if it has no nontrivial congruence.
When $\Theta$ is a congruence of $M$ and
$p~\Theta~q$ are states falling into the same $\Theta$-class, then $p\equiv q$.
Given $M$ and a congruence $\Theta$ on $M$, the \emph{factor automaton} of $M$ is
$M/\Theta~=~(Q/\Theta,\Sigma,\delta/\Theta,q_0/\Theta,F/\Theta)$ where
$p/\Theta$ denotes the $\Theta$-class of $p$, $X/\Theta$ denotes the set $\{p/\Theta:~p\in X\}$
of $\Theta$-classes for a set $X\subseteq Q$ and $\delta(q/\Theta,a)=\delta(q,a)/\Theta$ if
$\delta(q,a)$ is defined, and is undefined otherwise.

Then, for each $p\in Q$ the states $p$ and $p/\Theta$ are equivalent,
thus any automaton is equivalent to each of its factor automata.
It is also known that for any automaton $M$ there is a unique (up to isomorphism, i.e. modulo renaming states)
equivalent reduced automaton,
the one we get by trimming $M$, then factoring the useful part of $M$
by the language equivalence relation $p\Theta_Mq~\Leftrightarrow~p\equiv q$.

Any automaton can be seen as an edge-labeled multigraph and thus its \emph{strongly connected components},
or SCCs,
are well-defined classes of its states: the states $p$ and $q$ belong to the same SCC if $pu=q$ and $qv=p$
for some words $u,v\in\Sigma^*$. Clearly, this is an equivalence relation. We call an SCC \emph{trivial}
if it consists of a single state $p$ with $pu\neq p$ for any nonempty
word $u$ (that is, if it contains absolutely
no edges, not even loops), and \emph{nontrivial} otherwise.

\section{Reversible languages}
An automaton $M=(Q,\Sigma,\delta,q_0,F)$ is \emph{reversible} if $pa = qa$ implies $p = q$ for each $p,q\in Q$
and $a\in\Sigma$.
A language $L\subseteq\Sigma^*$ is \emph{reversible} if it is recognizable by some reversible automaton.
A \emph{reversible congruence} of a reversible automaton $M$ is a congruence $\Theta$
of $M$ such that the factor automaton $M/\Theta$ is also reversible.
The automaton $M$ is a \emph{reduced reversible automaton} if it has no nontrivial reversible congruence.

It is known~\cite{holzer15} that a language is reversible if and only if its minimal automaton
has no distinct states $p\neq q$, a letter $a$ and a word $w$ such that $pa=qa$ and $paw=w$
(the forbidden pattern is depicted on Figure~\ref{fig-rev-pattern}).
Equivalently, for any state $r$ belonging to a nontrivial component, and letter $a$,
the set $\{p\in Q:pa=r\}$ has to have at most one element.
\begin{center}
	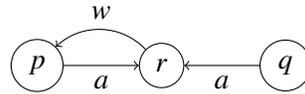
\begin{figure}[h]
		\centering\begin{tikzpicture}
		\node[draw,circle] (p) {$p$};
		\node[draw,circle,right=of p] (r) {$r$};
		\node[draw,circle,right=of r] (q) {$q$};
		\path[->]
		(p) edge [below] node [align=center] {$a$} (r)
		(q) edge [below] node [align=center] {$a$} (r)
		(r) edge [above,bend right=45] node [align=center] {$w$} (p)
		;
		\end{tikzpicture}			
	\caption{The forbidden pattern for reversible languages}
	\label{fig-rev-pattern}
	\end{figure}
\end{center}

Contrary to the general case of regular languages, there can be more than one reduced reversible automata,
up to isomorphism, recognizing the same (reversible) language. For example, see Figure~\ref{fig-harommajom} of~\cite{giovanni}.
\begin{center}
	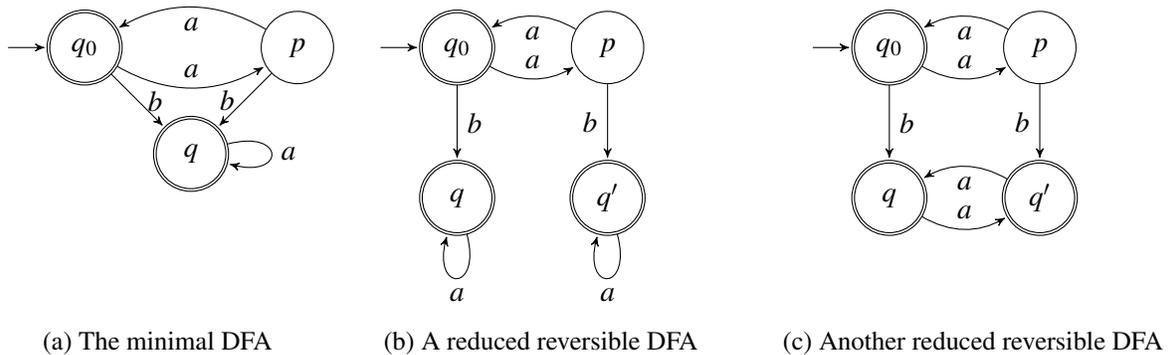
\begin{figure}[h!]
		\begin{subfigure}[b]{4.5cm}
			\centering\begin{tikzpicture}[>=stealth',shorten >=1pt,auto,node distance=2 cm, scale = 1, transform shape,initial text={}]
			\draw[white] (-1,-3.5) rectangle (3.5,1);
			\node[initial,state,accepting] (q0) {$q_0$};
			\node[state, accepting, below right of=q0] (q) {$q$};
			\node[state, above right of=q] (p) {$p$};
			\path[->]
			(q0) edge [above, bend right=30]    node [align=center]  {$a$} (p)
			(p) edge [below, bend right=30]      node [align=center]  {$a$} (q0)
			(q0) edge [right]       node [align=center]  {$b$} (q)
			(p) edge [left]      node [align=center]  {$b$} (q)
			(q) edge [loop right] 	node [align=center]  {$a$} (q)
			;			
			\end{tikzpicture}
			\caption{The minimal DFA}
		\end{subfigure}
		\begin{subfigure}[b]{5.5cm}
		\centering\begin{tikzpicture}[>=stealth',shorten >=1pt,auto,node distance=2 cm, scale = 1, transform shape,initial text={}]
			\draw[white] (-1,-3.5) rectangle (3.5,1);
			\node[initial,state,accepting] (q0) {$q_0$};
			\node[state, accepting, below of=q0] (q) {$q$};
			\node[state, accepting, right of=q] (qp) {$q'$};
			\node[state, above of=qp] (p) {$p$};
		\path[->]
		(q0) edge [above, bend right=30]    node [align=center]  {$a$} (p)
		(p) edge [below, bend right=30]      node [align=center]  {$a$} (q0)
		(q0) edge [right]       node [align=center]  {$b$} (q)
		(p) edge [left]      node [align=center]  {$b$} (qp)
		(q) edge [loop below] 	node [align=center]  {$a$} (q)
		(qp) edge [loop below] 	node [align=center]  {$a$} (qp)
		;
		\end{tikzpicture}
		\caption{A reduced reversible DFA}
	\end{subfigure}
		\begin{subfigure}[b]{5.8cm}
	\centering\begin{tikzpicture}[>=stealth',shorten >=1pt,auto,node distance=2 cm, scale = 1, transform shape,initial text={}]
			\draw[white] (-1,-3.5) rectangle (3.5,1);
	\node[initial,state,accepting] (q0) {$q_0$};
	\node[state, accepting, below of=q0] (q) {$q$};
	\node[state, accepting, right of=q] (qp) {$q'$};
	\node[state, above of=qp] (p) {$p$};
	\path[->]
	(q0) edge [above, bend right=30]    node [align=center]  {$a$} (p)
	(p) edge [below, bend right=30]      node [align=center]  {$a$} (q0)
	(q0) edge [right]       node [align=center]  {$b$} (q)
	(p) edge [left]      node [align=center]  {$b$} (qp)
	(q) edge [above, bend right=30] 	node [align=center]  {$a$} (qp)
	(qp) edge [below, bend right=30] 	node [align=center]  {$a$} (q)
	;
	\end{tikzpicture}
	\caption{Another reduced reversible DFA}
\end{subfigure}
		\caption{The case of the language $(aa)^*+a^*ba^*$}
		\label{fig-harommajom}
	\end{figure}
\end{center}
In that example the minimal automaton (depicted on Subfigure~(a)) is not reversible
and there are two nonisomorphic reversible reduced automata recognizing the same
language (with four states). We note that in this particular example there are
actually an infinite number of reduced reversible automata, recognizing the 
same language.
In~\cite{giovanni} the set of states of a minimal automaton was partitioned into
two classes: the \emph{irreversible} states are such states which are reachable
from some distinct states $p\neq q$ with the same word $w$, while the \emph{reversible}
states are those which are not irreversible. For example, in the case of Figure~\ref{fig-harommajom},
$q$ is the only irreversible state.
One of the results of~\cite{giovanni} is that if there exists an irreversible state $p$
which is reachable from a nontrivial SCC of the automaton (allowing the case when $p$
itself belongs to a nontrivial SCC, as $q$ does in the the example), then there exist
an infinite number of nonisomorphic reduced reversible automata, each recognizing the
same (reversible) language. A natural question is to precisely characterize the class of these
reversible languages.

\section{Result}
In this section we give a forbidden pattern characterization for those reversible languages
having a finite number of reduced reversible automata, up to isomorphism.

For this part, let us fix a reversible language $L$.
Let $M=(Q^*,\Sigma,\delta^*,q_0^*,F^*)$ be the minimal automaton of $L$.
We partition the states of $M$ into classes as follows: a state $q$ is a\ldots
\begin{itemize}
	\item \emph{$1$-state} if there exists exactly one word $u$ with $q_0^*u=q$;
	\item \emph{$\infty$-state} if there exist infinitely many words $u$ with $q_0^*u=q$;
	\item \emph{$\oplus$-state} if it is neither a $1$-state nor an $\infty$-state
\end{itemize}
	and orthogonally, $q$ is an\ldots
\begin{itemize}
	\item \emph{irreversible state} if there are distinct states $p_1^*\neq p_2^*\in Q^*$
	and a word $u$ such that $p_1^*u=p_2^*u=q$;
	\item \emph{reversible} if it is not irreversible.
\end{itemize}

As an example, consider Figure~\ref{fig-minimal}.
\begin{center}
	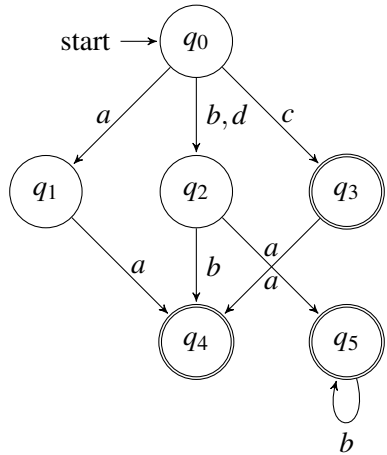
\begin{figure}[h]
	\centering\begin{tikzpicture}[>=stealth',shorten >=1pt,auto,node distance=2 cm, scale = 1, transform shape]
	\node[initial, state] (q0)  {$q_0$};
	\node[state] (q2)  [below of=q0]  {$q_2$};
	\node[state,accepting] (q3)  [right of=q2]  {$q_3$};
	\node[state] (q1)  [left of=q2]  {$q_1$};
	\node[state, accepting] (q4)  [below of=q2]  {$q_4$};
	\node[state, accepting] (q5)  [below of=q3]  {$q_5$};
	
	\path[->] (q0) edge [left]    node [align=center]  {$a$} (q1)
	(q0) edge [right]      node [align=center]  {$b,d$} (q2)
	(q0) edge [right]       node [align=center]  {$c$} (q3)
	(q1) edge [right]      node [align=center]  {$a$} (q4)
	(q2) edge [right] 	node [align=center]  {$b$} (q4)
	(q3) edge [above] 	node [align=center]  {$a$} (q4)
	(q2) edge [below] 	node [align=center]  {$a$} (q5)
	(q5) edge [loop below] node [align=center]  {$b$} (q5);
	\end{tikzpicture}
	\caption{The minimal automaton of our running example language}
	\label{fig-minimal}
\end{figure}
\end{center}
Here, states $q_0$, $q_1$ and $q_3$ are $1$-states, reachable by the words $\varepsilon$, $a$ and $c$,
respectively; $q_2$ and $q_4$ are $\oplus$-states as they are reachable by $\{b,d\}$ and $\{aa,bb,db,ca\}$,
respectively and $q_5$ is a $\infty$-state, reachable by words of the form $(b+d)ab^*$.
Moreover, $q_4$ is the only irreversible state (as $q_1a=q_3a$).

We note that our notion of irreversible states is not exactly the same as in~\cite{giovanni}:
what we call irreversible states are those states which belong to the ``irreversible part'' of
the automaton in the terms of~\cite{giovanni}. There, a state $q$ is called irreversible only 
if $p_1a=p_2a=q$ for some distinct pair $p_1\neq p_2$ of states and \emph{letter} $a$.

Clearly, a state is an $\infty$-state iff it can be reached from some nontrivial SCC of $M$.
Now we define the set $Z\subseteq Q^*$ of \emph{zig-zag states}\footnote{The coined term ``zig-zag'' originates from an earlier
version of Figure~\ref{fig-chain} on which forward edges had a ``northeast'' direction while backward edges had a ``southeast'' direction.} as follows: $Z$ is the least set $X$ satisfying
the following conditions:
\begin{enumerate}
	\item All the $\infty$-states belong to $X$.
	\item If $q\in X$ and $a\in\Sigma$ is a letter with $q\cdot a$ being defined, then $q\cdot a\in X$.
	\item If $q$ is a $\oplus$-state and $a\in \Sigma$ is a letter with $q\cdot a\in X$, then $q\in X$.
\end{enumerate}

The main result of the paper is the following:
\begin{theorem}
\label{thm-main}
	There are only finitely many reduced reversible automata recognizing $L$
	if and only if every zig-zag state of $M$ is reversible.
\end{theorem}
We break the proof up into several parts.
\subsection{When all the zig-zag states are reversible}
In this part we show that whenever all the zig-zag states are reversible,
and $N$ is a trim reversible automaton recognizing $L$, then there is a 
reversible congruence $\Theta$ on $N$ such that $N/\Theta$ has a bounded number of states (the bound in question is computable from the minimal automaton $M$).
So let us assume that there is no irreversible zig-zag state in $M$
and let $N=(Q,\Sigma,\delta,q_0,F)$ be a trim reversible automaton recognizing $L$.
Then, for each $q\in Q$ there exists a unique state $q^*\in Q^*$ with $q\equiv q^*$,
and the function $q\mapsto q^*$ is a homomorphism.

Now let us define the relation $\Theta$ on $Q$ as follows:
\[p\Theta q\quad\Leftrightarrow\quad (p=q)~\textrm{ or }~(p^*=q^*\in Z).\]
\begin{lemma}
	The relation $\Theta$ is a reversible congruence on $N$.
\end{lemma}
\begin{proof}
	It is clear that $\Theta$ is an equivalence relation: reflexivity and symmetry are trivial,
	and $p~\Theta~q~\Theta~r$ either entails $p=q$ or $q=r$ (in which case $p~\Theta~r$ is obvious)
	or that $p^*=q^*=r^*\in Z$ (and then, $p~\Theta~r$ also holds).
	
	Now if $p~\Theta~q$ and $p\cdot a$ is defined, then we have to show that $q\cdot a$
	is also defined and $p\cdot a~\Theta~q\cdot a$. This is clear if $p=q$. Otherwise
	we have $p^*=q^*$ is a zig-zag state,
	thus $(p\cdot a)^*=(q\cdot a)^*$ as starring is a homomorphism
	from $N$ to $M$ (thus in particular, $q\cdot a$ is defined).
	As $p^*\in Z$ and $Z$ is closed under action, we have that $p^*\cdot a=(p\cdot a)^*$ is
	also a zig-zag state, thus $p\cdot a~\Theta~q\cdot a$ indeed holds and $\Theta$ is
	a congruence on $N$.
	
	To see that $\Theta$ is a reversible congruence, assume $p\cdot a~\Theta~q\cdot a$.
	We have to show that $p~\Theta~q$. If $p\cdot a=q\cdot a$, then $p=q$ (thus $p~\Theta~q$),
	since $N$ is a reversible automaton. Otherwise, let $p\cdot a\neq q\cdot a$ 
	(hence $p\neq q$) and  $(p\cdot a)^*=(q\cdot a)^*\in Z$. By assumption on $M$,
	this state $(p\cdot a)^*$ is reversible. Thus, as $p^*\cdot a=q^*\cdot a=(p\cdot a)^*$,
	we get that $p^*=q^*$. Hence to show $p^*~\Theta~q^*$ it suffices to show that it is
	also a zig-zag state. If $p^*$ is a $\infty$-state, then it is a zig-zag state by
	definition of $Z$. Also, if $p^*$ is a $\oplus$-state, then it is still a zig-zag
	state (as $p^*\cdot a$ is a zig-zag state, we can apply Condition 3 in the definition
	of $Z$). Finally, observe that $p^*$ cannot be a $1$-state since $p$ and $q$ are
	different (reachable) states of $N$, hence there are words $u\neq v$ with 
	$q_0u=p$ and $q_0v=q$. For these words, by starring being a homomorphism we get that
	$q_0^*u=q_0^*v=p^*$($=q^*$), thus $p^*$ is reachable by at least two distinct words.
	
	Hence, $\Theta$ is indeed a reversible congruence.
\end{proof}
To conclude this case observe the following facts:
\begin{itemize}
	\item For each non-$\infty$-state $p^*$
		there exist a \emph{finite} number	$u_1,\ldots,u_{n(p^*)}$ of words leading into $p^*$ in $M$,
		thus there can be at most $n(p^*)$ states in $N$ which are equivalent to $p^*$ (since $N$ is
		trim). This bound $n(p^*)$ is computable from $M$.
	\item If $p^*$ is a $\infty$-state of $M$, then $p^*$ is a zig-zag state, thus
		all the states of $N$ equivalent to $p^*$ are collapsed into a single class of $\Theta$.
		For these states, let us define the value $n(p^*)$ to be $1$.
\end{itemize}
Hence, $n=\mathop\sum\limits_{p^*\in Q^*}n(p^*)$ is a (finite, computable) upper bound on the
number of states in the factor automaton $N/\Theta$ (hence it is an upper bound for the
number of states in any reduced reversible automaton recognizing $L$ as in that case $\Theta$
has to be the trivial congruence). Thus we have proved the first part of Theorem~\ref{thm-main}:
\begin{theorem}
\label{thm-revzigzag}
	If all the zig-zag states of $M$ are reversible, then
	there is a finite upper bound for the number of states
	of any reduced reversible automata recognizing $L$.
	Hence, in that case there exists only a finite number of nonisomorphic reduced
	reversible automata recognizing $L$.
\end{theorem}
\subsection{When there is an irreversible zig-zag state in $M$}
In this part let us assume that $M$ has an irreversible zig-zag state.
We will start from an arbitrary reduced reversible automaton $N$ recognizing $L$,
and then ``blow it up'' to some arbitrarily large equivalent reduced reversible automaton.
Before giving the construction, we illustrate the process in the case of Figure~\ref{fig-minimal}
(there, $q_4$ is an irreversible zig-zag state).
\begin{centering}\begin{figure*}[h]
	\begin{subfigure}{7cm}
		\begin{adjustbox}{raise=1.5cm}
			\begin{tikzpicture}[>=stealth',shorten >=1pt,auto,node distance=2 cm, scale = 1, transform shape]
			\node[initial, state] (q0)  {$q_0$};
			\node[state] (q2)  [below of=q0]  {$q_2$};
			\node[state,accepting] (q3)  [right of=q2]  {$q_3$};
			\node[state] (q1)  [left of=q2]  {$q_1$};
			\node[state, accepting] (q4)  [below of=q1]  {$q_4$};
			\node[state, accepting] (q41)  [below of=q3]  {$q_4'$};
			\node[state, accepting] (q5)  [below of=q2]  {$q_5$};
			
			\path[->] (q0) edge [left]    node [align=center]  {$a$} (q1)
			(q0) edge [right]      node [align=center]  {$b,d$} (q2)
			(q0) edge [right]       node [align=center]  {$c$} (q3)
			(q1) edge [right]      node [align=center]  {$a$} (q4)
			(q2) edge [right] 	node [align=center]  {$b$} (q4)
			(q3) edge [right] 	node [align=center]  {$a$} (q41)
			(q2) edge [right] 	node [align=center]  {$a$} (q5)
			(q5) edge [loop below] node [align=center]  {$b$} (q5);
			\end{tikzpicture}\end{adjustbox}
		\subcaption{Reversible automaton $N$ equivalent to $M$}
	\end{subfigure}
	\begin{subfigure}{10cm}
		\begin{tikzpicture}[>=stealth',shorten >=1pt,auto,node distance=2 cm, scale = 1, transform shape]
		\node[initial, state] (q0)  {$q_0$};
		\node[state] (q2)  [below left of=q0]  {$q_2$};
		\node[state] (q22)  [below right of=q0]  {$q_2'$};
		\node[state,accepting] (q3)  [right of=q22]  {$q_3$};
		\node[state] (q1)  [left of=q2]  {$q_1$};
		\node[state, accepting] (q4)  [below left of=q2]  {$q_4$};
		\node[state, accepting] (q42)  [below right of=q22]  {$q_4'$};
		\node[state, accepting] (q51)  [below of=q2]  {$q_{51}$};
		\node[state, accepting] (q52)  [below of=q22]  {$q_{52}$};
		\node[state, accepting] (q53)  [below of=q52]  {$q_{53}$};
		\node[state, accepting] (q54)  [below left of=q53]  {$q_{54}$};
		\node[state, accepting] (q55)  [below of=q51]  {$q_{55}$};

		\path[->] (q0) edge [left]    node [align=center]  {$a$} (q1)
		(q0) edge [right]      node [align=center]  {$b$} (q2)
		(q0) edge [left]      node [align=center]  {$d$} (q22)
		(q0) edge [right]       node [align=center]  {$c$} (q3)
		(q1) edge [right]      node [align=center]  {$a$} (q4)
		(q2) edge [right] 	node [align=center]  {$b$} (q4)
		(q22) edge [above] 	node [align=center]  {$b$} (q42)
		(q22) edge [right] 	node [align=center]  {$a$} (q52)
		(q3) edge [right] 	node [align=center]  {$a$} (q42)
		(q2) edge [right] 	node [align=center]  {$a$} (q51)
		(q51) edge [below] node [align=center]  {$b$} (q52)
		(q52) edge [right] node [align=center]  {$b$} (q53)
		(q53) edge [below] node [align=center]  {$b$} (q54)
		(q54) edge [below] node [align=center]  {$b$} (q55)
		(q55) edge [left] node [align=center]  {$b$} (q51)
		;
		\end{tikzpicture}
		\subcaption{The state $q_5$ is blown up by $5$, yielding $N'$}
	\end{subfigure}
	\caption{Blowing up a reversible automaton}
	\label{fig-blowing-example}
\end{figure*}
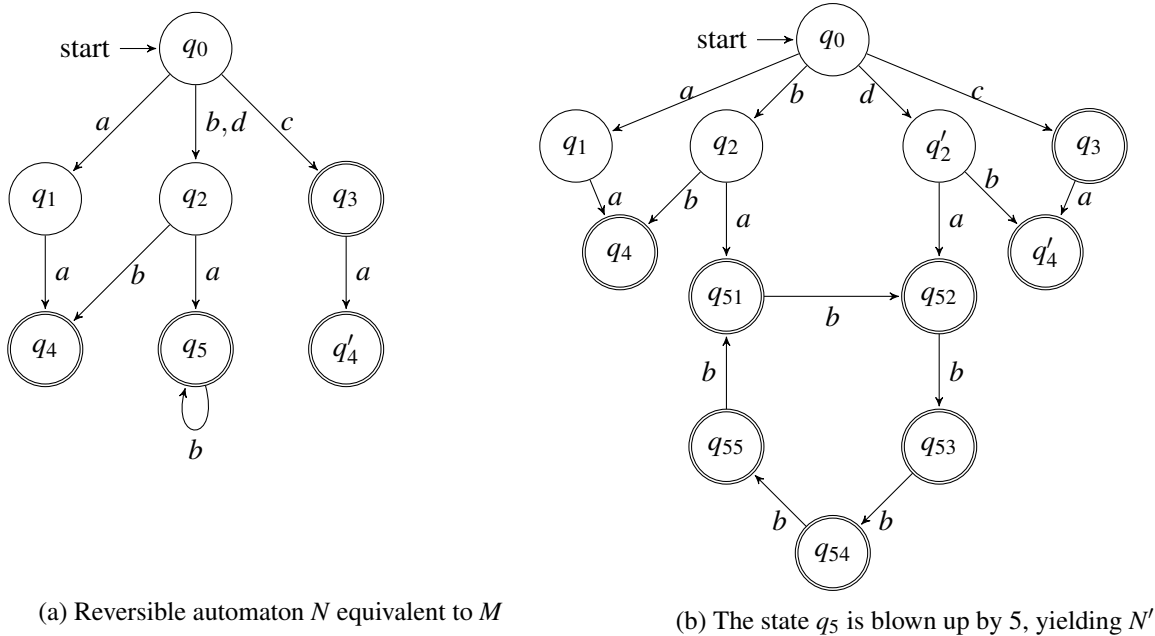
\end{centering}

On Figure~\ref{fig-blowing-example} (a), we have a reversible automaton $N$.
On (b), we replace the state $q_5$ having a loop by a cycle of length $5$
(and we also duplicate the state $q_2$ -- we can do that since $q_2$ is not a $1$-state).
Then, the automaton $N'$ is a reduced reversible automaton, whenever the length
of the cycle (which is now chosen to be $5$) is a prime number~\cite{giovanni}.

Indeed, any reversible congruence $\Theta$ on $N'$ collapses equivalent states only.
Suppose e.g. $\Theta$ collapses $q_{51}$ and $q_{53}$. Applying $b$ we get that 
$q_{52}$ and $q_{54}$ also get collapsed, and so $q_{53}$ and $q_{55}$,
and $q_{54}$ and $q_{51}$ -- thus, all the $q_{5x}$ states fall into a single $\Theta$-class
then. As $\Theta$ is assumed to be reversible, it has to collapse $q_2$ and $q_2'$
as well; applying $b$ also the states $q_4$ and $q_4'$ have to be collapsed
and finally, applying again reversibility we get that $q_1$ and $q_3$ should be
collapsed but this cannot happen as they are not equivalent states.
This reasoning works for any choice of different copies of $q_5$
(as far as the number chosen is a prime), thus $N'$ has
only the trivial reversible congruence and is a reduced reversible automaton.

The careful reader might realize that we actually followed in this reasoning
a zig-zag path from the $\infty$-state $q_5$ to the irreversible state $q_4$
during the above reasoning. In this part we show that this approach can always
be generalized whenever there exists an irreversible zig-zag state in $M$.

By the definition of the zig-zag states, if there exists some irreversible zig-zag state,
then there is a sequence
\[r_0,~(a_1,e_1),~r_1,~(a_2,e_2),~\ldots,~(a_\ell,e_\ell),~r_\ell\]
such that $a_i\in\Sigma$ and $e_i\in\{+,-\}$ for each $1\leq i\leq \ell$, and $r_i\in Q^*$ are states of $M$ for each $0\leq i\leq \ell$, moreover,
\begin{enumerate}[label=(\roman*)]
	\item if $e_i=+$ then $r_{i-1}\cdot a_i=r_i$ (denoted by $r_{i-1}\mathop{\longrightarrow}\limits^{a_i}r_i$ in the examples and patterns),
	\item if $e_i=-$ then $r_i\cdot a_{i}=r_{i-1}$ and $r_i$ is a $\oplus$-state (denoted by $r_{i-1}\mathop{\longleftarrow}\limits^{a_i}r_i$),
	\item $r_\ell$ is an irreversible state,
	\item $r_0$ is an $\infty$-state.
\end{enumerate}
Let us choose such a sequence of minimal length. Then, by minimality,
\begin{itemize}
	\item the states $r_i$ are pairwise different,
	\item all the states $r_i$, $1\leq i<\ell$ are reversible $\oplus$-states.
\end{itemize}
To see that all the states are $\oplus$-states (them being pairwise different reversible states is obvious),
observe that if $r_i$ is an $\infty$-state for $0<i$, then $r_i,\ldots,r_\ell$ is a shorter such sequence.
Hence all the states $r_i$, $0<i$ are either $1$-states or $\oplus$-states.
We show by induction that all of them are $\oplus$-states. The claim holds for 
$i=1$ as $e_1=+$ would imply that $r_1$ should be an $\infty$-state which cannot happen
thus $e_1=-$, hence $r_1$ is an $\oplus$-state. Now if $r_i$ is an $\oplus$-state,
then either $r_{i+1}=r_ia_{i+1}$ (if $e_{i+1}=+$) which implies that $r_{i+1}$
cannot be a $1$-state (thus it is an $\oplus$-state), or
$r_{i+1}$ is an $\oplus$-state (if $e_{i+1}=-$, applying ii)).

Now we extend the above sequence in both directions as follows.

First, $r_0$ being an $\infty$-state implies that there exists a state $p_0$
belonging to a nontrivial SCC of $M$ (that is, $p_0w=p_0$ for some nonempty
word $w$) from which $r_0$ is reachable. That is, there is a word
$b_1b_2\ldots b_m$ and states $p_1,\ldots,p_m$ with $p_m=r_0$
and $p_ib_{i+1}=p_{i+1}$ for each $0\leq i<m$.

Second, $r_\ell$ being an irreversible state implies that there exist
different states $s$ and $s'$ of $M$ from which 
$r_\ell$ is reachable by the same (nonempty) word.
That is, there is a word $c_1c_2\ldots c_n$ and states $s_1,s_2,\ldots,s_n=r_\ell$
such that $sc_1=s'c_1=s_1$ and $s_ic_{i+1}=s_{i+1}$ for each $1\leq i<n$.
See Figure~\ref{fig-chain}.
\begin{center}
	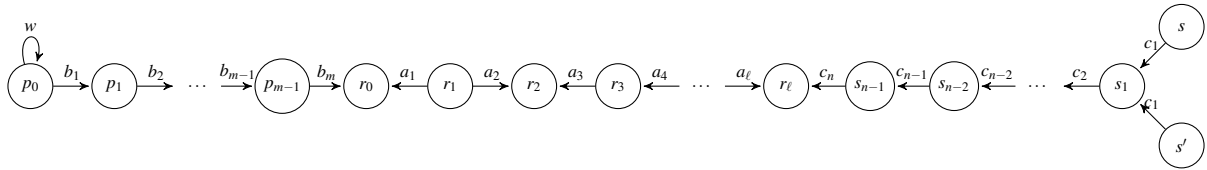
\begin{figure}
		\centering\begin{tikzpicture}[>=stealth',shorten >=1pt,auto,node distance=1.8 cm, scale = 0.62, transform shape]

		\node[state] (p0)  {$p_0$};
		\node[state] (p1) [right of=p0]{$p_1$};
		\node[state,draw=none] (pd) [right of=p1]{$\dots$};
		\node[state] (pm1) [right of=pd]{$p_{m-1}$};

		\node[state] (r0)  [right of=pm1]{$r_0$};
		\node[state] (r1) [right of=r0]{$r_1$};
		\node[state] (r2) [right of=r1]{$r_2$};
		\node[state] (r3) [right of=r2]{$r_3$};
		\node[state,draw=none] (rd) [right of=r3]{$\dots$};
		\node[state] (rl) [right of=rd]{$r_{\ell}$};

		\node[state] (sn1)  [right of=rl]{$s_{n-1}$};
		\node[state] (sn2)  [right of=sn1]{$s_{n-2}$};
		\node[state,draw=none] (sd) [right of=sn2]{$\dots$};
		\node[state] (s1) [right of=sd]{$s_1$};
		\node[state] (s) [above right of=s1]{$s$};
		\node[state] (sp) [below right of=s1]{$s'$};
		
		\path[->] 
		(p0) edge [loop above]    node [align=center]  {$w$} (p0)
		(p0) edge [above]    node [align=center]  {$b_1$} (p1)
		(p1) edge [above]      node [align=center]  {$b_2$} (pd)
		(pd) edge [above]      node [align=center]  {$b_{m-1}$} (pm1)
		(pm1) edge [above]      node [align=center]  {$b_m$} (r0)

		(r1) edge [above]    node [align=center]  {$a_1$} (r0)
		(r1) edge [above]      node [align=center]  {$a_2$} (r2)
		(r3) edge [above]      node [align=center]  {$a_3$} (r2)
		(rd) edge [above]      node [align=center]  {$a_4$} (r3)
		(rd) edge [above]      node [align=center]  {$a_\ell$} (rl)

		(s) edge [above]    node [align=center]  {$c_1$} (s1)
		(sp) edge [above]    node [align=center]  {$c_1$} (s1)
		(s1) edge [above]    node [align=center]  {$c_2$} (sd)
		(sd) edge [above]    node [align=center]  {$c_{n-2}$} (sn2)
		(sn2) edge [above]    node [align=center]  {$c_{n-1}$} (sn1)
		(sn1) edge [above]    node [align=center]  {$c_n$} (rl)
		;		
		\end{tikzpicture}
		\caption{The sequences $p_i$, $r_i$ and $s_i$.}
		\label{fig-chain}
	\end{figure}
\end{center}
Note that all the states $s_1,\ldots,s_n$ are $\oplus$-states and $p_0,\ldots,p_m$ are $\infty$-states.

In order to reduce the clutter in the notation, we treat the whole sequence $p_0,\ldots,p_{m-1},r_0,\ldots,s_1,s$ and $s'$
as a single indexed sequence $p_0,\ldots,p_{t-1},p_{t,1}$ and $p_{t,2}$ (see Figure~\ref{fig-chain-2}).
Observe that each $p_i$ (but possibly $p_{t,1}$ and $p_{t,2}$) is either an $\infty$- or a $\oplus$-state, and 
each of them is a zig-zag state. Moreover, all these states are pairwise different (thus inequivalent).
\begin{center}
	\begin{figure}
		\centering\begin{tikzpicture}[>=stealth',shorten >=1pt,auto,node distance=1.8 cm, scale = 0.62, transform shape]
		\node[state] (p0)  {$p_0$};
		\node[state] (p1) [right of=p0]{$p_1$};
		\node[state,draw=none] (pd) [right of=p1]{$\dots$};
		\node[state] (pm1) [right of=pd]{$p_{m-1}$};
		
		\node[state] (r0)  [right of=pm1]{$p_m$};
		\node[state] (r1) [right of=r0]{$p_{m+1}$};
		\node[state] (r2) [right of=r1]{$p_{m+2}$};
		\node[state] (r3) [right of=r2]{$p_{m+3}$};
		\node[state,draw=none] (rd) [right of=r3]{$\dots$};
		\node[state] (rl) [right of=rd]{$p_{m+\ell}$};
		
		\node[state] (sn1)  [right of=rl]{$p_{m+\ell+1}$};
		\node[state] (sn2)  [right of=sn1]{$p_{m+\ell+2}$};
		\node[state,draw=none] (sd) [right of=sn2]{$\dots$};
		\node[state] (s1) [right of=sd]{$p_{t-1}$};
		\node[state] (s) [above right of=s1]{$p_{t,1}$};
		\node[state] (sp) [below right of=s1]{$p_{t,2}$};
		
		\path[->] 
		(p0) edge [loop above]    node [align=center]  {$w$} (p0)
		(p0) edge [above]    node [align=center]  {$b_1$} (p1)
		(p1) edge [above]      node [align=center]  {$b_2$} (pd)
		(pd) edge [above]      node [align=center]  {$b_{m-1}$} (pm1)
		(pm1) edge [above]      node [align=center]  {$b_m$} (r0)
		
		(r1) edge [above]    node [align=center]  {$b_{m+1}$} (r0)
		(r1) edge [above]      node [align=center]  {$b_{m+2}$} (r2)
		(r3) edge [above]      node [align=center]  {$b_{m+3}$} (r2)
		(rd) edge [above]      node [align=center]  {$b_{m+4}$} (r3)
		(rd) edge [above]      node [align=center]  {$b_{m+\ell}$} (rl)
		
		(s) edge [above]    node [align=center]  {$b_{t}$} (s1)
		(sp) edge [above]    node [align=center]  {$b_{t}$} (s1)
		(s1) edge [above]    node [align=center]  {$b_{t-1}$} (sd)
		(sd) edge [above]    node [align=center]  {$b_{m+\ell+3}$} (sn2)
		(sn2) edge [above]    node [align=center]  {$b_{m+\ell+2}$} (sn1)
		(sn1) edge [above]    node [align=center]  {$b_{m+\ell+1}$} (rl)
		;		
		\end{tikzpicture}
		\caption{The sequence appearing in $M$, in an uniform notation}
		\label{fig-chain-2}
	\end{figure}
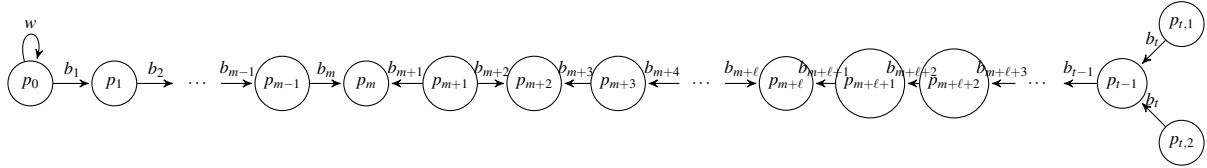
\end{center}

In the first step we show that if there is a specific pattern (which is a bit
more general than a cycle of prime length) appears in a reversible automaton $N'$,
then every factor automaton of $N'$ is ``large''.
\begin{lemma}
	\label{lem-pattern}
	Assume $N'$ is a reversible automaton, $k\geq 1$ is a prime number,
	$t\geq 0$ and $i<k$ are integers,
	$q_0,\ldots,q_{k-1}$, $p'_1,\ldots,p'_t$, $p''_1,\ldots,p''_t$ are states,
	$a_1,\ldots,a_t\in\Sigma$ are letters,
	$e_1,\ldots,e_t\in\{+,-\}$ are directions
	and $w\in\Sigma^+$ is a word
	satisfying the following conditions:
	\begin{itemize}
		\item $p'_t$ is not equivalent to $p''_t$,
		\item $q_jw=q_{j+1}$ for each $0\leq j<k$ with the convention that $q_{k}=q_0$,
		that is, indices of the $q$s are taken modulo $k$,
		\item for each $1\leq j\leq t$ with $e_j=+$ we have $p_{j-1}'a_j=p_j'$ and $p_{j-1}''a_j=p_j''$,
		\item and for each $1\leq j\leq t$ with $e_j=-$ we have $p_{j}'a_j=p_{j-1}'$ and $p_{j}''a_j=p_{j-1}''$
	\end{itemize}
	with setting $p_0':=q_0$ and $p_0'':=q_i$
	(See Figure~\ref{fig-pattern}).
	
	Then whenever $\Theta$ is a reversible congruence on $N'$, the
	states $q_j$ belong to pairwise different $\Theta$-classes.	
	(In particular, $N'/\Theta$ has at least $k$ states.)
\end{lemma}
\begin{center}
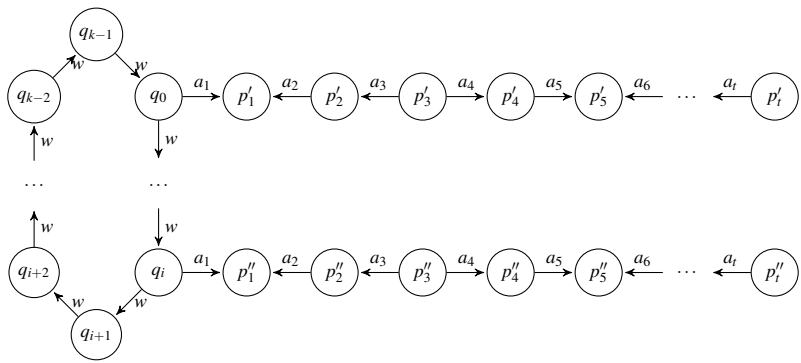
\begin{figure}
\centering\begin{tikzpicture}[>=stealth',shorten >=1pt,auto,node distance=1.8 cm, scale = 0.65, transform shape]
\node[state] (q0)  {$q_0$};
\node[state,below of=q0,draw=none] (qd) {$\ldots$};
\node[state,below of=qd] (qi) {$q_i$};
\node[state,below left of=qi] (qi1) {$q_{i+1}$};
\node[state,above left of=qi1] (qi2) {$q_{i+2}$};
\node[state,above of=qi2,draw=none] (qd2) {$\ldots$};
\node[state,above of=qd2] (qk2) {$q_{k-2}$};
\node[state,above right of=qk2] (qk1) {$q_{k-1}$};

\node[state] (p1) [right of=q0]{$p'_1$};
\node[state] (p2) [right of=p1]{$p'_2$};
\node[state] (p3) [right of=p2]{$p'_3$};
\node[state] (p4) [right of=p3]{$p'_4$};
\node[state] (p5) [right of=p4]{$p'_5$};
\node[state,draw=none] (pd) [right of=p5]{$\ldots$};
\node[state] (pt) [right of=pd]{$p'_t$};

\node[state] (p11) [right of=qi]{$p''_1$};
\node[state] (p21) [right of=p11]{$p''_2$};
\node[state] (p31) [right of=p21]{$p''_3$};
\node[state] (p41) [right of=p31]{$p''_4$};
\node[state] (p51) [right of=p41]{$p''_5$};
\node[state,draw=none] (pd1) [right of=p51]{$\ldots$};
\node[state] (pt1) [right of=pd1]{$p''_t$};

\path[->] 
(q0) edge [right] node [align=center] {$w$} (qd)
(qd) edge [right] node [align=center] {$w$} (qi)
(qi) edge [right] node [align=center] {$w$} (qi1)
(qi1) edge [right] node [align=center] {$w$} (qi2)
(qi2) edge [right] node [align=center] {$w$} (qd2)
(qd2) edge [right] node [align=center] {$w$} (qk2)
(qk2) edge [right] node [align=center] {$w$} (qk1)
(qk1) edge [right] node [align=center] {$w$} (q0)

(q0) edge [above]    node [align=center]  {$a_1$} (p1)
(qi) edge [above]    node [align=center]  {$a_1$} (p11)
(p2) edge [above]    node [align=center]  {$a_2$} (p1)
(p21) edge [above]    node [align=center]  {$a_2$} (p11)
(p3) edge [above]    node [align=center]  {$a_3$} (p2)
(p31) edge [above]    node [align=center]  {$a_3$} (p21)
(p3) edge [above]    node [align=center]  {$a_4$} (p4)
(p31) edge [above]    node [align=center]  {$a_4$} (p41)
(p4) edge [above]    node [align=center]  {$a_5$} (p5)
(p41) edge [above]    node [align=center]  {$a_5$} (p51)
(pd) edge [above]    node [align=center]  {$a_6$} (p5)
(pd1) edge [above]    node [align=center]  {$a_6$} (p51)
(pt) edge [above]    node [align=center]  {$a_t$} (pd)
(pt1) edge [above]    node [align=center]  {$a_t$} (pd1)
;		
\end{tikzpicture}
	\caption{The zig-zag pattern}
	\label{fig-pattern}	
\end{figure}
\end{center}

Before proceeding with the proof, the reader is encouraged to check that the
above described pattern appears in the automaton $N'$ of Figure~\ref{fig-blowing-example}
with the choice of $k=5$, $q_0=q_{51}$, $q_1=q_{52}$,\ldots,$q_4=q_{55}$,
$i=1$, $w=b$, $t=3$, $p_1'=q_2$, $p_1''=q_2'$, $p_2'=q_4$, $p_2''=q_4'$, $p_3'=q_1$,
$p_3''=q_3$, $a_1=a$, $a_2=b$, $a_3=a$, $e_1=-$, $e_2=+$ and $e_3=-$,
with states appearing
on the left-hand side of these equations are the states from the pattern of Lemma~\ref{lem-pattern}
while states on the right-hand side are states of $N'$.
\begin{proof}
	Assume for the sake of contradiction that $\Theta$ is a reversible congruence
	on $N'$ collapsing the states $q_\ell$ and $q_j$ for some $1\leq \ell<j\leq k$,
	that is, $q_\ell~\Theta~q_{j}$.
	
	We claim that $q_{\ell+d}~\Theta~q_{j+d}$ for each $d\geq 0$. This holds by
	assumption for $d=0$. Using induction on $d$, assuming $q_{\ell+d}~\Theta~q_{j+d}$
	we get by applying $w$ that $q_{\ell+d+1}=q_{\ell+d}w~\Theta~q_{j+d}w=q_{j+d+1}$
	as $\Theta$ is a congruence.
	Hence, writing $j=\ell+\delta$ we get that $q_{\ell+d}~\Theta~q_{\ell+\delta+d}$ for each $d\geq 0$,
	thus in particular for multiples of $\delta$: $q_{\ell+d\cdot \delta}~\Theta~q_{\ell+(d+1)\cdot \delta}$.
	
	Hence we have that $q_\ell~\Theta~q_{\ell+\delta}~\Theta~q_{\ell+2\delta}~\Theta\ldots$.
	As $k$ is assumed to be a prime number, there are integers $d_1$ and $d_2$
	with $\ell+d_1\delta\equiv 0~\mathrm{ mod }~k$ and $\ell+d_2\delta\equiv i~\mathrm{mod}~k$,
	thus $q_0~\Theta~q_i$. As $p_0'$ is defined as $q_0$ and $p_0''$ is defined as $q_i$,
	we have $p_0'~\Theta~p_0''$.
	
	Now for any integer $d\geq 0$, the relation $p'_{d}~\Theta~p''_d$ implies
	$p'_{d+1}~\Theta~p''_{d+1}$: if $e_{d+1}=+$, then by applying $a_{d+1}$ (since $\Theta$ is a congruence),
	while if $e_{d+1}=-$, then be reversibly applying $a_{d+1}$
	(since $\Theta$ is a reversible congruence). 
	Hence, it has to be the case $p'_t~\Theta~p''_t$
	which is nonsense since these two states are assumed to be inequivalent and $\Theta$ is a congruence.
\end{proof}
Observe that if some reversible automaton $N'$ recognizing $L$
admits the pattern of Lemma~\ref{lem-pattern} for some
prime number $k$, then there exists a reduced reversible automaton
of the form $N'/\Theta$ (also recognizing $L$) which then has at least $k$
states. 

In the rest of this part we show that if there exists an irreversible zig-zag
state in $M$, then we can construct such an automaton $N'$ for arbitrarily large
primes $k$, given a reversible automaton $N$ recognizing $L$.
Thus in that case it is clear that there exists an infinite number
of reduced reversible automata (up to isomorphism) recognizing $L$.

First we show that even a weaker condition suffices.
\begin{lemma}
	\label{lem-rewiring}
	Suppose $N$ is a reversible automaton recognizing $L$
	such that for each zig-zag state $p$ of $M$ there exist at least two states $p'$ and $p''$
	of $N$ with $p\equiv p'\equiv p''$ and to the $\infty$-state $p_0$ of $M$ (of Figure~\ref{fig-chain}), there
	exist at least $k$ different states $q_0,\ldots,q_{k-1}$ in $N$, each being equivalent to $p_0$,
	with $q_jw=q_{j+1}$ for each $0\leq j<k$ (again, with $q_k=q_0$).
	
	Then there exists a reversible automaton $N'$ also recognizing $L$ which admits the zig-zag pattern.
\end{lemma}
For an example reversible automaton $N$ recognizing $L$ but \emph{not} admitting the zig-zag pattern the
reader is referred to Figure~\ref{fig-notpattern}. We not prove that in these cases the transitions
can be ``rewired''.
\begin{center}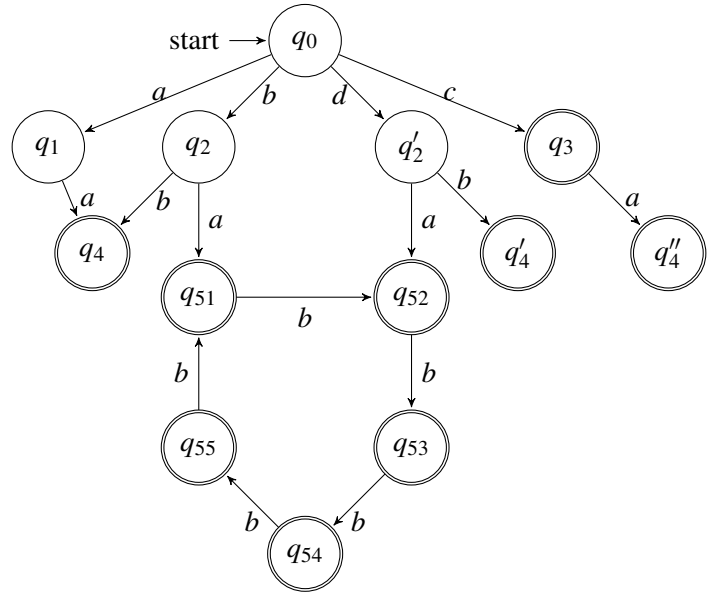
\begin{figure}
\centering\begin{tikzpicture}[>=stealth',shorten >=1pt,auto,node distance=2 cm, scale = 1, transform shape]
\node[initial, state] (q0)  {$q_0$};
\node[state] (q2)  [below left of=q0]  {$q_2$};
\node[state] (q22)  [below right of=q0]  {$q_2'$};
\node[state,accepting] (q3)  [right of=q22]  {$q_3$};
\node[state] (q1)  [left of=q2]  {$q_1$};
\node[state, accepting] (q4)  [below left of=q2]  {$q_4$};
\node[state, accepting] (q42)  [below right of=q22]  {$q_4'$};
\node[state, accepting] (q43)  [right of=q42]  {$q_4''$};
\node[state, accepting] (q51)  [below of=q2]  {$q_{51}$};
\node[state, accepting] (q52)  [below of=q22]  {$q_{52}$};
\node[state, accepting] (q53)  [below of=q52]  {$q_{53}$};
\node[state, accepting] (q54)  [below left of=q53]  {$q_{54}$};
\node[state, accepting] (q55)  [below of=q51]  {$q_{55}$};

\path[->] (q0) edge [left]    node [align=center]  {$a$} (q1)
(q0) edge [right]      node [align=center]  {$b$} (q2)
(q0) edge [left]      node [align=center]  {$d$} (q22)
(q0) edge [right]       node [align=center]  {$c$} (q3)
(q1) edge [right]      node [align=center]  {$a$} (q4)
(q2) edge [right] 	node [align=center]  {$b$} (q4)
(q22) edge [above] 	node [align=center]  {$b$} (q42)
(q22) edge [right] 	node [align=center]  {$a$} (q52)
(q3) edge [right] 	node [align=center]  {$a$} (q43)
(q2) edge [right] 	node [align=center]  {$a$} (q51)
(q51) edge [below] node [align=center]  {$b$} (q52)
(q52) edge [right] node [align=center]  {$b$} (q53)
(q53) edge [below] node [align=center]  {$b$} (q54)
(q54) edge [below] node [align=center]  {$b$} (q55)
(q55) edge [left] node [align=center]  {$b$} (q51)
;
\end{tikzpicture}
\caption{The automaton $N$ does not admit the zig-zag pattern}
\label{fig-notpattern}
\end{figure}
\end{center}

\begin{proof}
	We will construct a sequence $N=N_0$, $N_1$,\ldots, $N_t=N'$ of reversible automata,
	each recognizing $L$ (having the same set of states, and even $(N_i,p)\equiv(N_j,p)$ for each $i,j$ and $p$,
	that is, we do not change the languages recognized by any of the states of $N$)
	and sequences $p_0',\ldots,p_t'$ and $p_0'',\ldots,p_t''$ of states
	such that for each $0\leq j\leq t$ the following all hold:
	\begin{itemize}
		\item $p_0'=q_0$, $p_0''=q_1$,
		\item if $e_j=+$, then $p'_{j-1}b_j=p'_j$ and $p''_{j-1}b_j=p''_j$ for each automaton $N_\ell$ with $\ell\leq j$,
		\item if $e_j=-$, then $p'_{j}b_j=p'_{j-1}$ and $p''_{j}b_j=p''_{j-1}$ for each automaton $N_\ell$ with $\ell\leq j$,
		\item if $j<t$ and $\ell\leq j$, then $p_j$ is equivalent to both $p'_j$ and $p''_j$ in $N_\ell$,
		\item $p'_t\equiv p_{t,1}$ and $p''_t\equiv p_{t,2}$ in $N_t$.
	\end{itemize}
	If we manage to achieve this, then the automaton $N':=N_t$ indeed admits the zig-zag pattern and still recognizes $L$.
	
	We construct the above sequence $N_j$ by induction on $j$. For $j=0$, choosing $N_0=N$ and $p_0':=q_0$,
	$p_0'':=q_1$ suffices. Having constructed $N_j$, we construct $N_{j+1}$ based on whether the direction $e_{j+1}$
	is $+$ or $-$, the latter one having several subcases.
	\begin{enumerate}
		\item If $e_{j+1}=+$, then let us set $N_{j+1}:=N_j$, $p_{j+1}':=p_j'b_{j+1}$ (in $N_j$)
		and $p_{j+1}'':=p''_jb_{j+1}$ (also in $N_j$). This choice satisfies the conditions.
		\item If $e_{j+1}=-$, then we have three subcases, based on whether there exist states in $N_j$ equivalent to
			$p_{j+1}$ from which $b_{j+1}$ leads to either $p_j'$ or $p_j''$.
			\begin{enumerate}
				\item If there exist such states $r_1$ and $r_2$ with $r_1b_{j+1}=p_j'$ and $r_2b_{j+1}=p_j''$,
				then again, setting $N_{j+1}:=N_j$, $p'_{j+1}:=r_1$ and $p''_{j+1}:=r_2$ suffices.
				\item Assume there is no such $r_1$ nor $r_2$. Then, as $p_{j+1}$ is a zig-zag state, there exist,
				by the assumption on $N$, two different states $r_1$ and $r_2$, each being equivalent to $p_{j+1}$.
				Moreover, as $p_{j+1}b_{j+1}=p_j$ holds, we have that $r_1b_{j+1}$ and $r_2b_{j+1}$ (in $N_j$)
				are equivalent to $p_j$. So let us define $N_{j+1}$ as follows: $r_1b_{j+1}:=p'_j$, $r_2b_{j+1}:=p''_j$
				and for all the other pairs $(r,b)$ let us leave the transitions of $N_j$ unchanged.
				Then, setting $p'_{j+1}:=r_1$ and $p''_{j+2}:=r_2$ suffices.
				\item Finally, assume that exactly one of these predecessor states exists. By symmetry, we can assume
				that it is $r_1$, that is, $r_1b_{j+1}=p'_j$ in $N_j$ but there is no state $r$ equivalent to
				$p_{j+1}$ with $rb_{j+1}=p''_j$. Since $p_{j+1}$ is a zig-zag state, there exists some $r_2\neq r_1$
				in $N_j$, still equivalent to $p_{j+1}$. In this case we set the transitions in $N_{j+1}$
				as $r_2b_{j+1}p''_j$, leaving the other transitions unchanged suffices with $p'_{j+1}:=r_1$
				and $p''_{j+1}:=r_2$.
			\end{enumerate}
	\end{enumerate}
\end{proof}

Hence, given $N$, it suffices to construct a reversible automaton satisfying the conditions of Lemma~\ref{lem-rewiring}.
First we construct a reversible automaton $N'$ in which for each $\oplus$-state $p$ of $M$ there exist at least
two different equivalent states.
Let $U\subseteq\Sigma^*$ be the set of words $u$ 
such that $u=\varepsilon$ or $q_0u$ is \emph{not} a $\infty$-state.
Note that $U$ is
a nonempty finite set. Now let us define the state set of $N'$ as the finite set $Q'=Q\times U$, equipped with the following
transition function
\[(q,u)\cdot a~=~\begin{cases}
(qa,ua)&\hbox{if }qa\hbox{ is not an }\infty\hbox{-state},\\
(qa,u)&\hbox{otherwise}.
\end{cases}\]
For an illustration of the construction starting from $N$ of Figure~\ref{fig-blowing-example} (a), consult Figure~\ref{fig-wording}.
\begin{figure}[h]\centering
\begin{tikzpicture}[>=stealth',shorten >=1pt,auto,node distance=2 cm, scale = 1, transform shape]
\node[initial, state] (q0)  {$q_0,\varepsilon$};
\node[state] (q2)  [below left of=q0]  {$q_2,b$};
\node[state] (q2d)  [below right of=q0]  {$q_2,d$};
\node[state,accepting] (q3)  [right of=q2d]  {$q_3,c$};
\node[state] (q1)  [left of=q2]  {$q_1,a$};
\node[state, accepting] (q4)  [below of=q1]  {$q_4,bb$};
\node[state, accepting] (q4aa)  [left of=q4]  {$q_4,aa$};
\node[state, accepting] (q5)  [below of=q2]  {$q_5,b$};
\node[state, accepting] (q5d)  [below of=q2d]  {$q_5,d$};
\node[state, accepting] (q4db)  [right of=q5d]  {$q_4,db$};
\node[state, accepting] (q41)  [right of=q4db]  {$q_4',ca$};

\path[->] (q0) edge [left]    node [align=center]  {$a$} (q1)
(q0) edge [right]      node [align=center]  {$b$} (q2)
(q0) edge [left]      node [align=center]  {$d$} (q2d)
(q0) edge [right]       node [align=center]  {$c$} (q3)
(q1) edge [left]      node [align=center]  {$a$} (q4aa)
(q2) edge [right] 	node [align=center]  {$b$} (q4)
(q3) edge [right] 	node [align=center]  {$a$} (q41)
(q2) edge [right] 	node [align=center]  {$a$} (q5)
(q2d) edge [left] 	node [align=center]  {$a$} (q5d)
(q2d) edge [left] 	node [align=center]  {$b$} (q4db)
(q5) edge [loop below] node [align=center]  {$b$} (q5)
(q5d) edge [loop below] node [align=center]  {$b$} (q5d)
;
\end{tikzpicture}
\caption{All the $\oplus$-states (namely, $q_2$ and $q_4$) have several copies. (Only the trim part of the automaton is shown here.)}
\label{fig-wording}
\end{figure}
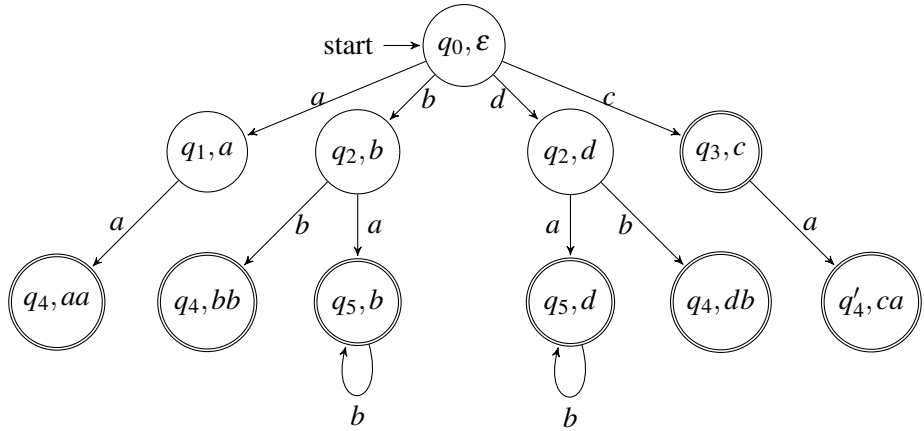

It is clear that $(q_0,\varepsilon)u=(q_0u,v)$ for some prefix $v$ of $u$, moreover, if $q_0u$ is a $\oplus$-state, then $v=u$.
Also, states of the form $(q,u)$ in $N'$ are equivalent to $q$ (if we set $F\times U$ as the accepting set).
Hence, whenever $q$ is a $\oplus$-state which can be reached by the words $u_1,\ldots,u_n$, $n\geq 2$,
then $(q,u_1),\ldots,(q,u_n)$ are pairwise different states in $N'$, reachable from $(q_0,\varepsilon)$ and thus to each $\oplus$-state $p$ of $M$
there exist at least two equivalent states in $N'$. Now starting from $N'$ we will construct an automaton $N''$ satisfying the conditions
of the Lemma.

Let $q$ be a state of $N'$, equivalent to $p_0$. Since $p_0w=p_0$ in $M$ for the nonempty word $w$,
the sequence $q$, $qw$, $qw^2$,\ldots contains some repetition. Let $i$ be the least integer with $qw^i=qw^j$ for some $j>i$.
Then if $i>0$, then we have $qw^{i-1}\cdot w=qw^i=qw^j=qw^{j-1}w$, thus $qw^{i-1}=qw^{j-1}$ since $N'$ is reversible.
Hence $q=qw^j$ for some integer $j$. In particular, $q$ belongs to a nontrivial SCC of $N'$.
Let $av$ be a shortest nonempty word with $qav=q$ (such a word exists since $q$ is in a nontrivial SCC).
For the fixed integer $k\geq 1$, let us define $N''$ as the automaton over the state set $Q'\times\{0,\ldots,k-1\}$,
with transition function
\[(q',i)b~=~\begin{cases}
(q'b,(i+1)~\textrm{mod}~k)&\hbox{if }q'=q\hbox{ and }b=a\\
(q'b,i)&\hbox{otherwise}.
\end{cases}\]
That is, we increase the index $i$ (modulo $k$) if we get the input $a$ in the state $q$, and in all the other cases the index
remains untouched.
For an example with $k=5$, see Figure~\ref{fig-circle}.
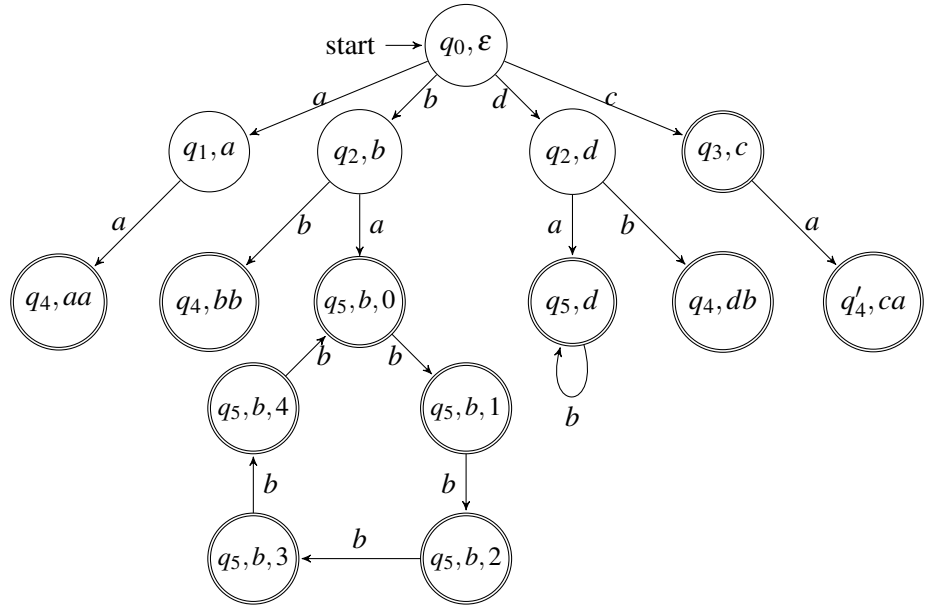
\begin{figure}[h]\centering
	\begin{tikzpicture}[>=stealth',shorten >=1pt,auto,node distance=2 cm, scale = 1, transform shape]
	\node[initial, state] (q0)  {$q_0,\varepsilon$};
	\node[state] (q2)  [below left of=q0]  {$q_2,b$};
	\node[state] (q2d)  [below right of=q0]  {$q_2,d$};
	\node[state,accepting] (q3)  [right of=q2d]  {$q_3,c$};
	\node[state] (q1)  [left of=q2]  {$q_1,a$};
	\node[state, accepting] (q4)  [below of=q1]  {$q_4,bb$};
	\node[state, accepting] (q4aa)  [left of=q4]  {$q_4,aa$};
	\node[state, accepting] (q5d)  [below of=q2d]  {$q_5,d$};
	\node[state, accepting] (q4db)  [right of=q5d]  {$q_4,db$};
	\node[state, accepting] (q41)  [right of=q4db]  {$q_4',ca$};
	\node[state, accepting,inner sep=2pt] (q50)  [below of=q2]  {\small$q_5,b,0$};
	\node[state, accepting,inner sep=2pt] (q51)  [below right of=q50]  {\small$q_5,b,1$};
	\node[state, accepting,inner sep=2pt] (q52)  [below of=q51]  {\small$q_5,b,2$};
	\node[state, accepting,inner sep=2pt] (q54)  [below left of=q50]  {\small$q_5,b,4$};
	\node[state, accepting,inner sep=2pt] (q53)  [below of=q54]  {\small$q_5,b,3$};

	\path[->] (q0) edge [left]    node [align=center]  {$a$} (q1)
	(q0) edge [right]      node [align=center]  {$b$} (q2)
	(q0) edge [left]      node [align=center]  {$d$} (q2d)
	(q0) edge [right]       node [align=center]  {$c$} (q3)
	(q1) edge [left]      node [align=center]  {$a$} (q4aa)
	(q2) edge [right] 	node [align=center]  {$b$} (q4)
	(q3) edge [right] 	node [align=center]  {$a$} (q41)
	(q2) edge [right] 	node [align=center]  {$a$} (q5)
	(q2d) edge [left] 	node [align=center]  {$a$} (q5d)
	(q2d) edge [left] 	node [align=center]  {$b$} (q4db)
	(q50) edge [left]node [align=center]  {$b$} (q51)
	(q51) edge [left]node [align=center]  {$b$} (q52)
	(q52) edge [above]node [align=center]  {$b$} (q53)
	(q53) edge [right]node [align=center]  {$b$} (q54)
	(q54) edge [right]node [align=center]  {$b$} (q50)
	(q5d) edge [loop below] node [align=center]  {$b$} (q5d)
	;
	\end{tikzpicture}
	\caption{State $(q_5,b)$ is blown up by a factor of $k=5$. (The states inequivalent to $q_5$ have the index equal to $0$
    which is not shown here.)}
	\label{fig-circle}
\end{figure}

As $av$ is a shortest word with $qav=q$, it is clear that $q$ does not occur on the $v$-path from $qa$ to $qav=q$.
Hence, in $N''$ we have $(q,i)av=(q,(i+1)~\mathrm{mod}~k)$. Moreover, if $u$ is a shortest word in $N'$ leading into $q$,
then it leads into $(q,0)$ in $N''$. Thus, in $N''$ we have the reachable states $q_0=(q,0)$, $q_1=(q,1)$, \ldots, $q_{k-1}=(q,k-1)$
with $q_iav=q_{(i+1)~\textrm{mod}~k}$ for each $i$, and still, for each $\oplus$-state $p$ there are at least two different states
in $N''$ equivalent to $p$. Hence, applying Lemma~\ref{lem-rewiring} we get that there exists some reversible automaton $N'''$
(note that the automaton $N''$ we constructed is also reversible) admitting the zig-zag pattern.

Figure~\ref{fig-rewired} shows the result of this last step: first, since $(q_5,b,0)$ has an incoming $a$-edge from $(q_2,b)$
but $(q_5,b,1)$ has no such edge, we search for another state equivalent to $q_2$, that's $(q_2,d)$. Then we set $(q_2,d)a$ to $(q_5,b,1)$.
Then, we can follow the $b$-transitions into $(q_4,bb)$ and $(q_4,db)$ respectively. After that, we should follow $a$-edges backwards
into $q_1$ and $q_3$. Hence we rewire the outgoing transitions as $(q_1,a)a=(q_4,bb)$ and $(q_3,c)a=(q_4,db)$ and all is set.
\begin{figure}[h]\centering
	\begin{tikzpicture}[>=stealth',shorten >=1pt,auto,node distance=2 cm, scale = 1, transform shape]
	\node[initial, state] (q0)  {$q_0,\varepsilon$};
	\node[state] (q2)  [below left of=q0]  {$q_2,b$};
	\node[state] (q2d)  [below right of=q0]  {$q_2,d$};
	\node[state,accepting] (q3)  [right of=q2d]  {$q_3,c$};
	\node[state] (q1)  [left of=q2]  {$q_1,a$};
	\node[state, accepting] (q4)  [below of=q1]  {$q_4,bb$};
	\node[state, accepting] (q4aa)  [left of=q4]  {$q_4,aa$};
	\node[state, accepting] (q5d)  [below of=q2d]  {$q_5,d$};
	\node[state, accepting] (q4db)  [right of=q5d]  {$q_4,db$};
	\node[state, accepting] (q41)  [right of=q4db]  {$q_4',ca$};
	\node[state, accepting,inner sep=2pt] (q50)  [below of=q2]  {\small$q_5,b,0$};
	\node[state, accepting,inner sep=2pt] (q51)  [below right of=q50]  {\small$q_5,b,1$};
	\node[state, accepting,inner sep=2pt] (q52)  [below of=q51]  {\small$q_5,b,2$};
	\node[state, accepting,inner sep=2pt] (q54)  [below left of=q50]  {\small$q_5,b,4$};
	\node[state, accepting,inner sep=2pt] (q53)  [below of=q54]  {\small$q_5,b,3$};

	\path[->] (q0) edge [left]    node [align=center]  {$a$} (q1)
	(q0) edge [right]      node [align=center]  {$b$} (q2)
	(q0) edge [left]      node [align=center]  {$d$} (q2d)
	(q0) edge [right]       node [align=center]  {$c$} (q3)
	(q1) edge [left]      node [align=center]  {$a$} (q4)
	(q2) edge [right] 	node [align=center]  {$b$} (q4)
	(q3) edge [right] 	node [align=center]  {$a$} (q4db)
	(q2) edge [right] 	node [align=center]  {$a$} (q5)
	(q2d) edge [left] 	node [align=center]  {$a$} (q51)
	(q2d) edge [left] 	node [align=center]  {$b$} (q4db)
	(q50) edge [left]node [align=center]  {$b$} (q51)
	(q51) edge [left]node [align=center]  {$b$} (q52)
	(q52) edge [above]node [align=center]  {$b$} (q53)
	(q53) edge [right]node [align=center]  {$b$} (q54)
	(q54) edge [right]node [align=center]  {$b$} (q50)
	(q5d) edge [loop below] node [align=center]  {$b$} (q5d)
	;
	\end{tikzpicture}
	\caption{The rewired automaton, admitting the zig-zag pattern. (States $(q_4,aa)$, $(q_5,d)$ and $(q_4',ca)$ are not part of the resulting trim automaton.)}
	\label{fig-rewired}
\end{figure}
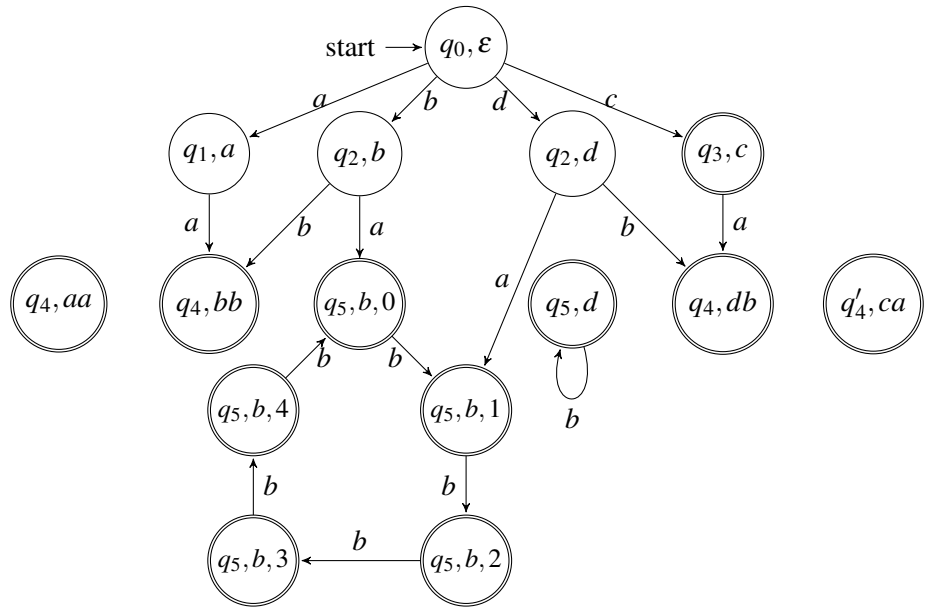

Thus, by Lemma~\ref{lem-pattern} we get the main result of the subsection:
\begin{theorem}
\label{thm-irrevzigzag}
	If there is an irreversible zig-zag state in $M$, then for an arbitrarily large $k$ one can effectively construct a
	reduced reversible automaton equivalent to $M$, having at least $k$ states.
\end{theorem}
Now Theorem~\ref{thm-main} is the conjunction of Theorems~\ref{thm-revzigzag}~and~\ref{thm-irrevzigzag}.
\section{Conclusion and acknowledgements}
We extended the current knowledge on the reversible regular languages by further analyzing the structure
of the minimal automaton of the language in question. In particular, we gave a forbidden pattern characterization
of those reversible languages having only a finite number of reduced reversible automata.
As the characterization relies on the existence of a “forbidden pattern” (that of Figure~\ref{fig-chain-2}), it gives an
efficient decision procedure, namely an $\mathbf{NL}$ (nondeterministic logspace) algorithm: one has to guess a
state $p_0$, then guess some loop from $p_0$ to itself, then following some back-and-forth walk in the graph
of the automaton to two distinct states $p_{t,1}$ and $p_{t,2}$. In the process we also have to check that no $1$-state
is encountered during this walk (which can clearly also be done in $\mathbf{NL}$).
It can be an interesting question to study the notion of reduced reversible automata in other reversibility
settings, as e.g. in the case of~\cite{lombardy}.
The authors wish to thank Giovanni Pighizzini, Giovanna Lavado and Luca Prigioniero for their
useful comments on a much earlier version of this paper.
\bibliographystyle{eptcs}
\bibliography{generic}
\end{document}